\documentclass{llncs} 

\usepackage[]{stix} 
\usepackage{amsmath} 
\usepackage{mathtools} 
\usepackage{graphicx} 
\usepackage{subcaption}
\usepackage[linesnumbered]{algorithm2e}
\usepackage{stmaryrd} 
\usepackage{listings}
\usepackage{tikz} 
\usetikzlibrary {positioning}

\usepackage[english]{babel}

\newcommand{\iden}[1]{{\text{\usefont{U}{dsrom}{m}{n}1}}_{#1}}
\newcommand{\ident}{\iden{}}
\newcommand{\expr}[1]{\mathbb{E}_{#1}}
\newcommand{\exprL}{\expr{\mathcal{L}}}
\newcommand{\converse}{^{\mkern-1mu{}{\raise0.3ex\hbox{\tiny$\smallsmile$}}}\kern-0.1em{}}
\newcommand{\sem}[2]{\left\lBrack#1\right\rBrack_{#2}}
\usepackage{url}

\newcommand{\WPC}[0]{Weak Pushout Chain}
\newcommand{\Wpc}[0]{Weak pushout chain}
\newcommand{\wpc}[0]{weak pushout chain}
\newcommand{\WPS}[0]{Weak Pushout Step}
\newcommand{\wps}[0]{weak pushout step}
\newcommand{\tr}[0]{sentence}
\newcommand{\Tr}[0]{Sentence}
\newcommand{\TR}[0]{Sentence}
\newcommand{\gr}[0]{graph rule}

\newcommand{\GR}[0]{Graph Rule}
\newcommand{\transL}[0]{\mathcal{E}_\mathcal{L}}

\tikzset{vertex/.style={draw,shape=circle,fill,inner sep=1pt}} 

\allowdisplaybreaks

\lstdefinestyle{custom}{
  belowcaptionskip=1\baselineskip,
  xleftmargin=0pt,
  basicstyle=\footnotesize\ttfamily,
  commentstyle=\itshape,
  columns=flexible,
}

\begin{document}

\title{Finding models through graph saturation} 
\author{Sebastiaan J. C. Joosten
}
\institute{Formal Methods and Tools group, University of Twente, the Netherlands \\
Email: \texttt{Sebastiaan.Joosten@utwente.nl}
}
\maketitle

\begin{abstract}
We give a procedure that can be used to automatically satisfy invariants of a certain shape.
These invariants may be written with the operations intersection, composition and converse over binary relations, and equality over these operations.
We call these invariants \tr{}s that we interpret over graphs.
For questions stated through sets of these sentences, this paper gives a semi-decision procedure we call graph saturation.
It decides entailment over these \tr{}s, inspired on graph rewriting.
We prove correctness of the procedure.
Moreover, we show the corresponding decision problem to be undecidable.
This confirms a conjecture previously stated by the author~\cite{amperspiegelRAMICS}.
\end{abstract}

This is an accepted preprint, which will be published in the Journal of Logical and Algebraic Methods in Programming (JLAMP).

\section{Introduction}

The question `what models does a set of formulas $\mathcal{T}$ have' has practical relevance, as it is an abstraction of an information system:
We interpret the data set stored in an information system at a certain point in time as a model, and each invariant of the system corresponds to a formula in $\mathcal{T}$.
This correspondence is the core idea behind languages such as Ampersand~\cite{AmpersandRAMICS}, that define an information system this way.
Users of an information system try to change the data set continually.
These changes might violate the constraints.
While Ampersand responds to such violations by rejecting the change, it would be convenient to automatically add data items such that all constraints are satisfied.
The question then becomes: what data items should be added?
We solve this question partially by means of a graph saturation procedure.

The question `does a set of formulas $\mathcal{T}$ have a model satisfying all formulas' essentially asks whether $\mathcal{T}$ is free of contradictions.
So far, we did not discuss the language in which we can write the formulas in $\mathcal{T}$.
Several interesting problems arise when restricting the language in which we can write formulas:
the satisfiability problem is obtained by restricting to disjunctions of positive and negative literals.
Restricting to linear integer equalities, we obtain the linear programming problem.
In this paper, we restrict those formulas to equalities over terms, in which terms are expressions of relations combined through the allegorical operations\footnote{These are $\fcmp$, $\sqcap$, $\converse{}$ and $\ident$. See the book by Freyd and Scedrov for details on allegories~\cite{FS90}.}.
We define \emph{\tr{}} to be a formula over the restricted language considered in this paper (Definition~\ref{def:tr}).

Our interest in this language stems from experience in describing systems in Ampersand.
All operations from relation-algebra are part of the Ampersand language.
The operations considered here include only the most frequently used subset of those operations.
Therefore, many of the formulas used in Ampersand will be sentences as considered in this work.
We therefore consider this work a step towards an Ampersand system that helps the user find models.


\subsection{Approach}
We give a short summary of the basic algorithm presented here, so we can better relate our approach to other literature, describe our contributions, and give the outline of this paper.
Italicised words in the next paragraph are defined later.

The algorithm aims to determine whether there is a particular \emph{model} for a set of \emph{\tr{}s}, say $\mathcal{T}$, and is guaranteed to terminate if no such \emph{model} exists.
It proceeds to construct a (possibly infinite) \emph{model} otherwise.
The procedure has two phases:
first, we translate the \emph{\tr{}s} in $\mathcal{T}$ into a set of \emph{\gr{}s}.
We then apply a saturation procedure on the \emph{\gr{}s}.
This procedure creates a \emph{chain} of \emph{graphs}, whose limit is a \emph{least consequence graph}.
A \emph{graph} contains a \emph{conflict} if it has an edge with the label $\bot{}$.
If a \emph{least consequence graph} contains a \emph{conflict}, then there is no \emph{model} for $\mathcal{T}$.
Otherwise, the \emph{least consequence graph} corresponds to a \emph{model} of $\mathcal{T}$, if the \emph{\gr{}s} correspond to $\mathcal{T}$ according to a straightforward \emph{translation}.
We abort the procedure as soon as a \emph{conflict} arises, because we can be sure that no \emph{models} for $\mathcal{T}$ exist in this case.
A second question we can answer through the same algorithm is that of \emph{entailment}:
\emph{entailment} is the question whether a \emph{\tr{}} $\phi$ follows from a set of \emph{\tr{}s} $\mathcal{T}$.

In an information system, a least consequence graph is a well suited to determine which data items to add:
If conflict free, it corresponds to a graph that maintains the invariants.
At the same, it only contains necessary consequences: it will not cause data items to be added that have nothing to do with the change the user made.

\subsection{Related Work}
We compare the work in this paper to existing work in two ways: work it is similar to in motivation, and work it is similar to in implementation from an abstract perspective.
In motivation, our research is closely related to the Alcoa tool, which we'll discuss first.
In approach, our methods are related to description logics and to graph rewriting, which we'll discuss second.

\paragraph{The Alcoa Tool.}
Our search for a reasoner for Ampersand is related to Alcoa~\cite{Jackson2000AlcoaTA}, which is the analyzer for Alloy~\cite{jackson2002alloy}, a language based on Z~\cite{Z}.
Like Ampersand, the languages Z and Alloy are based on relations.
Alloy is a simplification of Z: it reduces the supported operations to a set that is small yet powerful.
This paper differs from Alloy in the expressivity of its operations, however:
Alloy allows writing full first order formula's plus the Kleene-star, making it a language that is even more expressive than Ampersand.
We compare to Alcoa because this work is similar in purpose.

In Alloy, a user may write assertions, which are formulas that the user believes follow from the specification.
Alcoa tries to find counterexamples to those assertions, as well as a finite model for the entire specification.
Unfortunately, several properties of the Alcoa tool hinder our purposes in Ampersand:
Alcoa requires an upper bound on the size of (or number of elements in) the model.
It does not perform well if this bound is too large.
In a typical information system, the amount of data is well above what can be considered `too large'.
As an additional complication, we cannot adequately predict the size of the model we might require.
This is why we look at other methods for achieving similar goals.


\paragraph{Description Logics.}
We can regard our procedure as a way to derive facts from previously stated facts:
this is what happens in terms of \tr{}s between subsequent graphs in the chain we create.
So called description logics are languages used in conjunction with an engine, that gives a procedure to learn new facts from previously learned facts, using declarative statements (or rules) in the corresponding description logic.
For a good overview of description logics, see the book on that topic by Baader~\cite{baader2003description}.

A set of derivation rules is consistent if it has a model.
For a highly expressive description logic such as OWL DL, determining consistency is undecidable.
Still, a rule engine for OWL DL will happily try to learn new facts until a model is found.
Users of OWL DL typically need to ensure that the stated derivation rules together with the rule engine give a terminating procedure.
For many description logics, termination of its rule engine is syntactically guaranteed, and these logics are consequentially decidable.

The description logic for which the language and implementation is closest to our language is the logic $\mathcal{EL}$ and its extensions proposed by Baader et al~\cite{Baader:2005aa,EL++}.
Instead of using tableau-based procedures, as most description logics, it uses a saturation-based reasoner.
Syntax of the derivation rules is limited to ensure termination of any saturation procedure:
$\mathcal{EL}$ allows statements about unary relations using top, bottom, individual elements called `nominal', and conjunction.
Statements about binary relations use a different syntax, that can be translated into \tr{}s using composition, converse and the identity relation (but not necessarily vice-versa).
By modeling $\mathcal{EL}$'s unary relations as binary relations that are a subset of the identity relation, all of $\mathcal{EL}$ and its extensions can be expressed through the \tr{}s described in this paper.
In particular, the syntax of $\mathcal{EL}$ does not have disjunctions, thus eliminating the need for backtracking.
In fact, $\mathcal{EL}$ is designed such that its consistency can be decided deterministically in polynomial time.
Its extensions have different complexity bounds, but preserve polynomial runtime for the fragment that falls within $\mathcal{EL}$.

In our work, we do not work under the assumption of termination: neither the user or the syntax guarantees it.
This allows us to use a richer language than one that is syntactically guaranteed to terminate.
Despite this lack of termination, we do ensure termination in case of conflicts: a conflict will be found if our \tr{}s imply it.
This allows the user to approach certain problems through any set of rules within the grammar, rather than just those sets for which the implementation is guaranteed to terminate.
The implementation presented in our work applies \gr{}s `fairly' to ensure this.
Fair application of rules is typically not required in the implementation of description logic engines.

\paragraph{Graph Rewriting.}
A central concept in graph rewriting is that a pushout can be used to apply a \emph{\gr{}} on a graph, as described by Wolfram Kahl~\cite{kahl2001relation}.
The usual idea of such a pushout is that it models execution by removing a portion of the graph, and replacing it with the result of the execution step.
Graph rewriting might then terminate when no rules can be applied anymore.
Our approach diverges on this point: rather than execution, a step models learning a deducible conclusion.
Rather than terminating when no step is possible, we are interested in the limit of the sequence of graphs.
For this reason, the notions of weak pushout step and weak pushout don't coincide exactly:
we ensure that the sequence of graphs form a chain in order for the limit to exist.

The term saturation is borrowed from the saturation procedure in resolution procedures, introduced by Robinson in 1965~\cite{Robinson:1965aa}.
His procedure solves an entailment problem over a certain language.
As in his procedure, our procedure adds derivable facts iteratively.

\subsection{Contributions and Paper Outline}
We mentioned how this paper contributes by comparing it to related work:
Compared to the work on $\mathcal{EL}$, our approach allows \tr{}s in a richer language, and we present a translation to \gr{}s to separate the semantics from the core of the implementation.
Compared to the work on graph-rewriting, we present a new graph-based manipulation algorithm, and give an interpretation of those graphs as models for sets of \tr{}s.

We also relate the contribution of this paper to a paper presenting Amperspiegel~\cite{amperspiegelRAMICS}.
This earlier paper by the author conjectured that the problem whether no least consequence graph exists is undecidable.
It also contains a procedure for finding such graphs, which it conjectures to be correct.
We will show that the procedure in the paper is an instance of the variations of the procedure described here.
To simplify the presentation of our results in this paper, the definition of a least consequence graph is slightly different here:
A least consequence graphs always exist according to the definitions used in this paper.
In the terminology of this paper, the conjecture just mentioned would be: the problem whether a least consequence graph contains a conflict is undecidable.
This paper proves the stated result.

The procedure presented in this work is simpler than the one presented earlier.
However, the latter can be obtained by applying optimizations to the former.
We show correctness of the procedure, and show that the existence of a conflict free least consequence graph implies the existence of models for a set of sentences.
Semi-decidability of consistency is not surprising in this setting:
the logic we consider is less expressive than several logics for which semi-decidability is established.
Our contribution lies in presenting an intuitive, flexible, graph-based algorithm that does not use backtracking.

The outline of this paper is as follows:
we define the syntax and semantics of \tr{}s in Section~\ref{sec:background}, and define the problems our procedure aims to solve: deciding consistency and entailment.
Section~\ref{sec:graphrules} then introduces the heart of the procedure by defining least consequence graphs and indicating how to obtain them through \gr{}s.
Section~\ref{sec:Rules} connects these two, by giving a translation of \tr{}s to \gr{}s.
The procedure is given as an algorithm in Section~\ref{sec:procedure}, and we indicate how to use the procedure to decide consistency and entailment.
Before going to the conclusion, we indicate why we cannot hope to do better than giving a possibly non-terminating procedure, by proving undecidability in Section~\ref{sec:undecidable}.
Conclusion and acknowledgements are in Section~\ref{sec:conclusion}.

\section{Background and Problem Statement}\label{sec:background}
As this paper primarily deals with directed labeled graphs, we choose to use these graphs for the semantics of \tr{}s as well.
There is no fundamental difference between this presentation and the usual binary relation based semantics usually presented as the canonical allegory (or as the canonical model for relation algebra).
However, using graphs now simplifies our proofs later on, and makes it that we do not have to define them later.
Graphs are defined as follows:

\begin{definition}[Graph, Empty, Finite]
A directed labeled graph $G = (\mathcal{L},V,E)$ is given by a set of labels $\mathcal{L}$, a set of vertices $V$, and a set of edges $E \subseteq \mathcal{L} \times V \times V$.
The set of all graphs with labels $\mathcal{L}$ is written as $\mathbb{G}_\mathcal{L}$.
We write \emph{graph} when we mean a directed labeled graph.
We say that a graph is \emph{finite} if both its set of vertices $V$ and its set of edges $E$ are finite.
The cardinality of $V$ is written $|G|$.
A graph with no vertices (and therefore no edges) is called \emph{empty}, written $\mathbb{0}_\mathcal{L}$.
\end{definition}

Terms are built inductively from relation symbols $\mathcal{L}$, combined with the operations \ $\_\sqcap\_$, \ $\_\fcmp\_$, \ and \ $\_\converse$.
The operations stand for intersection, relational composition, and relational converse, respectively.
The set of all terms over $\mathcal{L}$ is denoted as $\exprL$.
We use the same letter $\mathcal{L}$ to indicate labels in graphs, as well as relation symbols in terms.
This notation is deliberately chosen because of the semantics given in Definition~\ref{def:semantics} below.
\begin{definition}[Semantics]\label{def:semantics}
For a graph $G=(\mathcal{L},V,E)$, the \emph{semantics} of a term $\mathbb{e} \in \exprL$, written as $\sem{\mathbb{e}}G \subseteq V \times V$, is as in representable relation algebra:
\begin{align*}
\sem{l}G &= \left\{(x,y) \mid (l,x,y) \in E \right\} \\
\sem{\mathbb{e}_1 \sqcap \mathbb{e}_2}G &= \sem{\mathbb{e}_1}G \cap \sem{\mathbb{e}_2}G \\
\sem{\mathbb{e} \converse{}\,}G &= \left\{ (y,x) \mid (x,y)\in\sem{\mathbb{e}}G\right\} \\
\sem{\mathbb{e}_1 \fcmp \mathbb{e}_2}G &= {\left\{ (x,y) \mid \exists z.\ (x,z)\in\sem{\mathbb{e}_1}G \ \wedge \ (z,y)\in\sem{\mathbb{e}_2}G\right\}}
\end{align*}
\end{definition}

A \tr{} is the proposition stating that two terms are equal:
\begin{definition}[\Tr{}, Holds]\label{def:tr}
Given the terms $\mathbb{e}_1,\mathbb{e}_2 \in \exprL$, the pair $(\mathbb{e}_1,\mathbb{e}_2)$ is a \emph{\tr{}}, written $\mathbb{e}_1 = \mathbb{e}_2$.
We write $\mathbb{e}_L \sqsubseteq \mathbb{e}_R$ for a \tr{} of the shape $\mathbb{e}_L = \mathbb{e}_L \sqcap \mathbb{e}_R$.
We say that a \tr{} \emph{holds} in graph $G$ if $\sem{\mathbb{e}_1}G = \sem{\mathbb{e}_2}G$, in which case we write: $G \vDash \mathbb{e}_1 = \mathbb{e}_2$.
If $\mathcal{T}$ is a set of sentences, we say that it holds in $G$ if each of the sentences holds in $G$, written $G \vDash \mathcal{T}$.
\end{definition}
\begin{lemma}\label{lem:subseteq}
Let $\mathbb{e}_1,\mathbb{e}_2 \in \exprL$, and $G\in\mathbb{G}_\mathcal{L}$.
\begin{align*}
G \vDash \mathbb{e}_1 \sqsubseteq \mathbb{e}_2 \quad &\Leftrightarrow \quad \sem{\mathbb{e}_1}G \subseteq \sem{\mathbb{e}_2}G\\
G \vDash \mathbb{e}_1 = \mathbb{e}_2 \quad &\Leftrightarrow \quad G \vDash \mathbb{e}_1 \sqsubseteq \mathbb{e}_2 \ \wedge \ G \vDash \mathbb{e}_2 \sqsubseteq \mathbb{e}_1
\end{align*}
\end{lemma}

We deviate slightly from allegories:
First, we are working in an untyped setting, or put differently: in an allegory with only a single object.
In `typed allegories', allegories with more than one object, relational composition is a partial operation.
This deviation is not fundamental: we are simply adding more terms to our language than would be there in the typed setting.
A second deviation is that we have not introduced identity morphisms.
We introduce the identity symbol $\ident{}$ by treating it as a symbol in $\mathcal{L}$.
Our approach generalizes to multiple identity symbols, as one would expect in allegories with multiple objects, but this is out of scope in favor of a simplified presentation.

Apart from the identity symbol $\ident{}$, we also introduce bottom and top ($\bot$ and $\top$) as symbols in $\mathcal{L}$.
In Definition~\ref{def:standard} we give the interpretation of these designated relation symbols, defining a graph as standard if it adheres to this interpretation.

\begin{definition}[Standard]\label{def:standard}
We say that a set of labels $\mathcal{L}$ is \emph{standard} with the (possibly empty) set of constant elements $\mathcal{C}$ if $\bot,\top,\ident \in \mathcal{L}$ and $\mathcal{C}\subseteq\mathcal{L}$.
We refer to elements in $\mathcal{C}$ simply as constants.
Let $\mathcal{L}$ be a standard set of labels with the constants $\mathcal{C}$.
A graph $G = (\mathcal{L},V,E)$ is called \emph{standard} if $V\neq \{\}$, and:
\begin{align*}
\sem{\ident}{G} &= \{(x,x) \mid x\in V\} \\
\sem{\top}G &= \{(x,y) \mid x,y\in V\} \\
\sem{\bot}G &= \{\} \\
\forall c\in\mathcal{C}.\ \sem{c}G &= \{(c,c)\}
\end{align*}
\end{definition}

This work looks at models for $\mathcal{T}$, and investigates whether $\mathcal{T}$ entails $\phi$.
We can now give the definitions that necessary to make this precise.
\begin{definition}[Model, Consistent]\label{def:model}
Let $\mathcal{T}$ be a set of \tr{}s over a standard set of labels $\mathcal{L}$ (with constants $\mathcal{C}$).
We say that the graph $G\in \mathbb{G}_\mathcal{L}$ is a \emph{model} for $\mathcal{T}$ if every \tr{} in $\mathcal{T}$ holds in $G$ and $G$ is standard.
We say that $\mathcal{T}$ is \emph{consistent} if such a graph exists.
We may refer to any set of \tr{}s $\mathcal{T}$ as an instance of the consistency problem.
\end{definition}
\begin{definition}[Entails]\label{def:entails}
Let $\mathcal{T}$ be a set of \tr{}s over a standard set of labels $\mathcal{L}$, and let $\phi$ be a \tr{} over $\mathcal{L}$.
We say that $(\mathcal{T},\phi)$ is an instance of the \emph{entailment problem}.
We say that $\mathcal{T}$ \emph{entails} $\phi$ if for all standard graphs G, $G \vDash \mathcal{T}$ implies $G \vDash \phi$.
\end{definition}

Our use of `standard' in these definitions is not a restriction:
given a graph $G$ over a language $\mathcal{L}$ with $\bot,\top,\ident \not\in \mathcal{L}$, we can make it into a standard graph $G'$ over $\mathcal{L} \cup \{\bot,\top,\ident\}$, choosing the constants $\mathcal{C}=\{\}$, and adding the edges according to Definition~\ref{def:standard}.
Then $G \vDash \phi$ if and only if $G' \vDash \phi$ for $\phi \in \exprL$, as $\phi$ cannot talk about $\bot,\top$ or $\ident$.

We prove a straightforward correspondence between the consistency problem and the entailment problem:
\begin{lemma}
There is a standard graph $G$ such that $G \vDash \mathcal{T}$ if and only if $\mathcal{T}$ does not entail $\bot = \top$.
\end{lemma}
\begin{proof}
We first prove that if $\mathcal{T}$ entails $\bot = \top$, then there is no standard graph with $G \vDash \mathcal{T}$:
A standard graph must have at least one vertex, say $v$.
Then $(v,v)\in\sem{\top}G$, and $(v,v)\not\in\sem{\bot}G$, so $\sem{\bot}G \neq \sem{\top}G$.
For the other direction:
Suppose there is no standard graph with $G \vDash \mathcal{T}$, then entailment of any formula follows by definition.\qed
\end{proof}

We proceed with a small example of sentences, an entailment and a consistency problem.
As an example, we make an administration of people and rooms.
We use the label ${\tt i}$ to denote which room a person Inhabits, and ${\tt r}$ to denote which people are Roommates.
We think of the labels in terms of their semantics: as binary relations.
We show how these relations are connected by the sentence expressing:
Two people are roommates if and only if they share a room: ${\tt r} = {\tt i}\fcmp{\tt i}\converse{}$.
This gives a one-sentence theory $\mathcal{T} = \{{\tt r} = {\tt i}\fcmp{\tt i}\converse{}\}$ on a standard set of labels that contains $\tt i$ and $\tt r$.

We ask ourselves if being a roommate is a transitive relation.
That is, does $\mathcal{T}$ entail ${\tt r}\fcmp{\tt r} \sqsubseteq {\tt r}$ or not?
The answer is negative.
A possible counter-example our procedure may produce is a graph $G$ with:
\begin{align*}
\sem{\tt i}G &= \{(0,3),(1,4),(2,3),(2,4)\} \\
\sem{\tt r}G &= \{(0,0),(0,2),(1,1),(1,2),(2,0),(2,1),(2,2)\}
\end{align*}
In this example, $0,1,2$ are people, and $3,4$ are their rooms.
While $0$ and $1$ are roommates of $2$, $0$ is not a roommate of $1$.
Note that person $2$ has two rooms in this example.
We may wish to forbid this: the sentence ${\tt i}\converse{} \fcmp {\tt i} \sqsubseteq \ident{}$ expresses that $\tt i$ is univalent (if two rooms are inhabited by the same person, those two must be the same room).
Now $\mathcal{T} = \{{\tt r} = {\tt i}\fcmp{\tt i}\converse{}\ ,\ {\tt i}\converse{} \fcmp {\tt i} \sqsubseteq \ident{}\}$ entails ${\tt r}\fcmp{\tt r} \sqsubseteq {\tt r}$, and our procedure shows this, as we will demonstrate in Section~\ref{sec:procedure}.

We elaborate on the same example for checking consistency, and add some constants to $\mathcal{L}$.
Let $\mathcal{C} = \{\texttt{`Liz'},\texttt{`Jon'},\texttt{`Batcave'},\texttt{`Room 11'}\}$.
Let's say we want Liz and Jon to be roommates, and ask ourselves if that's possible.
That is, we wish to solve the consistency problem for:
\begin{align*}
\mathcal{T} = \{\ & {\tt r} = {\tt i}\fcmp{\tt i}\converse{}\\
,\ &{\tt i}\converse{} \fcmp {\tt i} \sqsubseteq \ident{}\\
,\ &\texttt{`Liz'} \fcmp \top \fcmp \texttt{`Jon'} \sqsubseteq {\tt r} \}
\end{align*}
Our procedure then produces a graph like $G$ with:
\begin{align*}
\sem{\tt i}G &= \{(\texttt{`Liz'},0),(\texttt{`Jon'},0)\} \\
\sem{\tt r}G &= \{(\texttt{`Liz'},\texttt{`Jon'})\}
\end{align*}
Without going into details on why, we remark that our procedure comes up with a new room, here called $0$, even with the Batcave and Room 11 available.
Finally, if we require Liz and Jon to be in their rooms of their choice, the Batcave and Room 11 respectively, our procedure detects that the requirements are no longer consistent.
That is, the following theory is not consistent:
\begin{align*}
\mathcal{T} = \{\ & {\tt r} = {\tt i}\fcmp{\tt i}\converse{}\ ,\ {\tt i}\converse{} \fcmp {\tt i} \sqsubseteq \ident{}\\
,\ &\texttt{`Liz'} \fcmp \top \fcmp \texttt{`Jon'} \sqsubseteq {\tt r}\\
,\ &\texttt{`Liz'} \fcmp \top \fcmp \texttt{`Batcave'} \sqsubseteq {\tt i}\\
,\ &\texttt{`Jon'} \fcmp \top \fcmp \texttt{`Room 11'} \sqsubseteq {\tt i}
\}
\end{align*}

\section{\GR{}s and Consequence Graphs}\label{sec:graphrules}
This section defines a least consequence graph, and gives conditions on a chain of graphs that ensure that its limit is a least consequence graph.
When a graph is a least consequence graph, we can use it to answer both the entailment problem and the consistency problem.
The conditions on a chain of graphs tell us how \gr{}s should be applied by possible implementations.
Basically `least consequence graph' characterises that all \gr{}s are applied correctly and sufficiently.
We define \gr{}s in Definition~\ref{def:gr}, least consequence graphs in Definition~\ref{def:lcg}, and conclude the section with the conditions that give us a least consequence graph, proven in Lemma~\ref{lem:fairchain} and~\ref{lem:pushoutchain}.

We introduce special notation for two basic operations on graphs: relabeling of vertices, and taking the union of two graphs.
Suppose we have a function $f : V_1 \to V_2$, where $V_1$ is the set of vertices of some graph.
We can apply the function on the corresponding graph, written $\hat{f}$:
\[ \hat{f}((\mathcal{L},V_1,E)) = \left(\mathcal{L},\{f(v) \mid v\in V_1\},\{(l,f(x),f(y)) \mid (l,x,y) \in E\}\right)\]
For taking the union of two graphs, we simply write $\cup$, defined as follows:
\[(\mathcal{L}_1,V_1,E_1) \cup (\mathcal{L}_2,V_2,E_2) = (\mathcal{L}_1 \cup \mathcal{L}_2,V_1 \cup V_2,E_1 \cup E_2)\]
This leads to a natural definition of subgraph:

\begin{definition}[Subgraph]
We say that $G_1$ \emph{is a subgraph of} $G_2$ if $G_1 \cup G_2 = G_2$.
It follows that a subgraph of a finite graph is again finite.
If $G_1$ is a subgraph of $G_2$ and $G_1,G_2\in\mathbb{G}_\mathcal{L}$ for some $\mathcal{L}$, we write $G_1 \xrightarrow{\subseteq} G_2$.
\end{definition}
In this article, we consider the set of labels $\mathcal{L}$ to be arbitrary but fixed.
The relation `subgraph' forms a complete lattice over $\mathbb{G}_\mathcal{L}$, which justifies the following definition:
\begin{definition}[Chain, Supremum]
Given a set of labels $\mathcal{L}$.
We say that $S : \mathbb{N} \rightarrow \mathbb{G}_\mathcal{L}$ is a \emph{chain} if for all $i \in \mathbb{N}$, $S(i)$ is a subgraph of $S(i+1)$.
The union of all graphs in a chain, written $S(\infty)$, is called the \emph{supremum}, defined as $S(\infty) = (\mathcal{L},\bigcup_{i} E_i,\bigcup_{i} V_i)$ with $S(i) = (\mathcal{L},E_i,V_i)$.
\end{definition}

The way we use graph rewriting is most closely related to the single-pushout rewriting found in the literature (e.g.~\cite{kahl2001relation}).
In this approach, \gr{}s are related through a morphism that is, for instance, a partial function.
Vertices in the left hand side of the rule not related to the right hand side get removed upon application of the rule.
Similarly, vertices on the right hand side get inserted.
In our setting, we need the application of a rule to form a chain.
To make sure that we can do this, we use `subgraph' as a condition on \gr{}s.
\begin{definition}[\GR{}]\label{def:gr}
A pair of graphs $(L,R)$ is called a \emph{\gr{}} if $L$ is a subgraph of~$R$, and $R$ is finite.
We say that a set $\mathcal{R}$ is a \emph{set of \gr{}s with labels} $\mathcal{L}$ if each $(L,R)\in\mathcal{R}$ is a \gr{}, and $L,R\in \mathbb{G}_\mathcal{L}$.
\end{definition}

We proceed by giving an example of a \gr{}, and do so visually.
A graph can be drawn in the usual way.
Figure~\ref{fig:graph} is an example of a graph with $k,l,m\in\mathcal{L}$.
A picture does not specify the set of labels $\mathcal{L}$, only the set of edges and the set of vertices.
An example of a \gr{} is given in Figure~\ref{fig:rulpair} and~\ref{fig:rulcompact}.
Using the subgraph condition allows us to draw a \gr{} $(L,R)$ in a single figure, using small dots for nodes in $R$ but not in $L$, and big dots and solid edges for what is in $L$, and therefore in $R$.
\begin{figure}[btp]
\centering
\begin{subfigure}[b]{0.2\textwidth}\centering
\begin{tikzpicture}[-latex,semithick]
\node[vertex,label=above:$a$](1) {}; 
\node[vertex,label=above:$b$](2) [right =of 1] {};
\node[vertex,label=below:$c$](3) [below =of 2] {};
\path (1) edge[bend left] node[label=above:$k$] {} (2);
\path (3) edge[bend left] node[label=below:$m$] {} (1);
\path (2) edge node[label=right:$l$] {} (3);
\path (1) edge [loop left,distance=1cm,in=135,out=-135] node[label=$l$] {} (1);
\end{tikzpicture}
\caption{Graph $G_\mathrm{\ref{fig:graph}}$}\label{fig:graph}
\end{subfigure}
\begin{subfigure}[b]{0.45\textwidth}\centering
$\left(\begin{array}{c}
\begin{tikzpicture}[-latex,semithick]
\node[](c) {};  
\node[vertex,label=above:$0$](0) [left = 10pt of c] {};
\node[vertex,label=above:$1$](1) [right = 10pt of c] {};
\path (0) edge node[label=above:$l$] {} (1);
\end{tikzpicture}
\end{array},
\begin{array}{c}
\begin{tikzpicture}[-latex,semithick]
\node[](c) {};
\node[vertex,label=above:$0$](0) [left = 10pt of c] {};
\node[vertex,label=above:$1$](1) [right = 10pt of c] {};
\node[vertex,label=below:$2$](2) [below =of c] {};
\path (0) edge node[label=above:$l$] {} (1);
\path (0) edge node[label=below:$l$] {} (2);
\path (2) edge node[label=below:$l$] {} (1);
\end{tikzpicture}\end{array}
\right)$
\caption{A \gr{} $(L,R)$}\label{fig:rulpair}
\end{subfigure}
\begin{subfigure}[b]{0.3\textwidth}\centering
\begin{tikzpicture}[-latex,thin]
\node[](c) {};
\node[vertex,label=above:$0$,inner sep=2pt](0) [left = 10pt of c] {};
\node[vertex,label=above:$1$,inner sep=2pt](1) [right = 10pt of c] {};
\node[vertex,label=below:$2$](2) [below =of c] {};
\path (0) edge[thick] node[label=above:$l$] {} (1);
\path (0) edge[dotted] node[label=below:$l$] {} (2);
\path (2) edge[dotted] node[label=below:$l$] {} (1);
\end{tikzpicture}
\caption{$(L,R)$ compactly}\label{fig:rulcompact}
\end{subfigure}
\\
\begin{subfigure}[b]{0.2\textwidth}\centering
\begin{tikzpicture}[-latex,semithick]
\node[vertex,label=above:$0$](0) {};
\node[vertex,label=above:$1$](1) [right =of 0] {};
\path (0) edge node[label=above:$l$] {} (1);
\path (0) edge[loop below,distance=1cm,in=-135,out=-45] node[label=above left:$l$\ \ ] {} (0);
\end{tikzpicture}
\caption{Graph $G_\mathrm{\ref{fig:consequence1}}$}\label{fig:consequence1}
\end{subfigure}
\begin{subfigure}[b]{0.45\textwidth}\centering
\begin{tikzpicture}[-latex,semithick]
\node[vertex,label=above:$0$](0) {};
\node[vertex,label=above:$1$](1) [right =of 0] {};
\path (0) edge node[label=above:$l$] {} (1);
\path (1) edge[loop below,distance=1cm,out=-135,in=-45] node[label=above right:\ \ $l$] {} (1);
\end{tikzpicture}
\caption{Consequence graph $G_\mathrm{\ref{fig:consequence2}}$}\label{fig:consequence2}
\end{subfigure}
\begin{subfigure}[b]{0.3\textwidth}\centering
\begin{tikzpicture}[-latex,semithick]
\node[vertex,label=above:$0$](0) {};
\node[vertex,label=above:$1$](1) [right =of 0] {};
\path (0) edge node[label=above:$l$] {} (1);
\path (0) edge[loop below,distance=1cm,in=-135,out=-45] node[label=above left:$l$\ \ ] {} (0);
\node[vertex,label=above:$2$](0) [below =of 0]{}; 
\node[vertex,label=above:$3$](1) [right =of 0] {};
\path (0) edge node[label=above:$l$] {} (1);
\path (0) edge[loop below,distance=1cm,in=-135,out=-45] node[label=above left:$l$\ \ ] {} (0);
\end{tikzpicture}
\caption{Unconnected graph $G_\mathrm{\ref{fig:consequence3}}$}\label{fig:consequence3}
\end{subfigure}
\caption{Graphs and \gr{}s}
\end{figure}

We present a saturation procedure, so we need to capture when a graph is `saturated'.
For this purpose, we define `maintained', which indicates that a rule is applied sufficiently in a graph.
For defining `maintained', we first define graph embeddings:
\begin{definition}[Embedding]
Let $G_1,G_2\in\mathbb{G}_\mathcal{L}$.
If $\hat{f}(G_1) \xrightarrow{\subseteq} G_2$, then $(f,G_1,G_2)$ is an \emph{embedding} of $G_1$ in $G_2$.
In such case, we write $G_1 \xrightarrow{f} G_2$.
We say that $G_1$ is \emph{embedded} in $G_2$ if such an $f$ exists, written $G_1 \rightarrow G_2$.
It follows immediately that $G \xrightarrow{f} \hat{f}(G)$.
\end{definition}

We briefly explain our notations with the observation that embeddings form a category: its objects are graphs with labels $\mathcal{L}$, and its arrows are embeddings.
Although $\_ \xrightarrow{\subseteq} \_ = \_ \xrightarrow{\lambda x . x} \_ $, note that $G_1 \xrightarrow{\subseteq} G_2$ is only the identity arrow if $G_1 = G_2$,
which is why we avoid writing $\_ \xrightarrow{id}\_ $.

\begin{definition}[Maintained, (Least) Consequence Graph]\label{def:lcg}
A \gr{} $(L,R)$ with $L = (\mathcal{L},V_L,E_L)$ is \emph{maintained} in $G$ if for every embedding $L\xrightarrow{f}G$, there is an embedding $R\xrightarrow{g}G$ such that $f(v) = g(v)$ for all $v \in V_L$.
If for a set of \gr{}s $\mathcal{R}$, each \gr{} in $\mathcal{R}$ is maintained in $G$, we say that $G$ is a \emph{consequence graph} maintaining $\mathcal{R}$.
If furthermore $S$ is a subgraph of $G$, and $(S,G)$ is maintained in each consequence graph maintaining $\mathcal{R}$, then $G$ is a \emph{least consequence graph} of $S$ maintaining $\mathcal{R}$.
\end{definition}

We use chains to find least consequence graphs.
We look at two properties: `fairness' and `weak pushout', that help establish graphs to be a consequence graph and least, respectively.
To get some intuition, and hopefully help dispel some overly optimistic conjectures, we look at some examples before defining these two properties.

We begin with an example of an embedding.
Let $L = (\{k,l,m\},\{0,1\},\{(l,0,1)\})$ and $R = (\{k,l,m\},\{0,1,2\},\{(l,0,1),(l,0,2),(l,2,1)\})$ be graphs.
Note that $(L,R)$ is the \gr{} drawn in Figure~\ref{fig:rulpair}.
We can embed $L$ into the graph $G_\mathrm{\ref{fig:graph}}$ as shown in Figure~\ref{fig:graph}.
A corresponding embedding is $(f,L,G_\mathrm{\ref{fig:graph}})$ with $f(i) = a$ for $i\in\{0,1\}$.
There is also an embedding for $R$: $(g,R,G_\mathrm{\ref{fig:graph}})$ with $g(i) = a$ for $i\in\{0,1,2\}$, which satisfies $g(i)=f(i)$ for $i\in\{0,1\}$.
However, the \gr{} is not maintained, as for the embedding $(f',L,G_\mathrm{\ref{fig:graph}})$ with $f'(0) = b,\ f'(1) = c$, there is no such $g$.

As an example of a consequence graph, let $\mathcal{R} = \{(L,R)\}$ with $(L,R)$ as defined above, and let $G_\mathrm{\ref{fig:consequence1}} = (\{k,l,m\},\{0,1\},\{(l,0,1),(l,0,0)\})$, as drawn in Figure~\ref{fig:consequence1}.
Then $G_\mathrm{\ref{fig:consequence1}}$ is a consequence graph maintaining $\mathcal{R}$.
It is, however, not a least consequence graph of $L$ maintaining $\mathcal{R}$, since Figure~\ref{fig:consequence2} gives a consequence graph maintaining $\mathcal{R}$ in which $(L,G)$ is not maintained.
We believe every least consequence graph of $L$ maintaining $\mathcal{R}$ is infinite and even infinitely branching: loops in such consequence graphs would make that they are no longer `least', and to every edge with label $l$ there need to be two more of such edges in order to maintain $\mathcal{R}$.

The graph $G_\mathrm{\ref{fig:consequence1}}$ defined above is an example of a least consequence graph of $G_\mathrm{\ref{fig:consequence1}}$ maintaining $\mathcal{R}$.
Graph $G_\mathrm{\ref{fig:consequence3}}=(\{k,l,m\},\{0,1,2,3\},\{(l,0,1),(l,0,0),(l,2,3),(l,2,2)\})$, consisting of two disjunctive copies of $G_\mathrm{\ref{fig:consequence1}}$, is a least consequence graph too, see Figure~\ref{fig:consequence3}.
If a least consequence graph is unique, it must be the empty graph.

From the definition of maintained it follows that if $L \xrightarrow{\subseteq} M \xrightarrow{\subseteq} R$ and $(L,R)$ is maintained in $G$, then $(L,M)$ is maintained in $G$ too.
Consequently, if $(L,R)$ is maintained in a least consequence graph of $L$ maintaining $\mathcal{R}$, then $(L,R)$ is maintained in every consequence graph maintaining $\mathcal{R}$.

The following definition gives a sufficient condition to reach a consequence graph:
\begin{definition}[Fair Chain]
Given a set of \gr{}s $\mathcal{R}$ and a chain $S$.
We say that $S$ is a \emph{fair chain} for $\mathcal{R}$
if for each \gr{} $(L,R)\in \mathcal{R}$ and for each embedding $L\xrightarrow{f}S(i)$ there exists a $j\in\mathbb{N}$ and an embedding $R\xrightarrow{g}S(j)$ with $f(v) = g(v)$ for all $v$ in the set of vertices of $L$.
\end{definition}
\begin{lemma}\label{lem:fairchain}
If $S$ is a fair chain for $\mathcal{R}$, $S(\infty)$ is a consequence graph maintaining $\mathcal{R}$.
\end{lemma}
\begin{proof}
By definition, $S(\infty)$ is a consequence graph if we can show that $R$ is embedded in $S(\infty)$ for every $L$ that is embedded in it.
Take such an embedding $L\xrightarrow{f}S(\infty)$.
Then for each edge $(l,u,v)$ of $L$ there is an $i$ such that $(l,f(u),f(v))$ is an edge in $S(i)$.
Take the largest such $i$, then $f$ embeds $L$ in $S(i)$, and therefore $g$ embeds $R$ in $S(j)$ for some $j$ with $f(v) = g(v)$, so $g$ also embeds $R$ in $S(\infty)$.\qed
\end{proof}

We define \wps{} as an upper limit to each step, to ensure that a consequence graph found as a supremum of a chain built out of these steps is also a least consequence graph.
\begin{definition}[\WPS{}]
Let $G_1$ and $G_2$ be graphs in $\mathbb{G}_\mathcal{L}$, and let $(L,R)$ be a \gr{}.
We say that $(G_1,G_2)$ is a \emph{\wps{}} for $(L,R)$ if the following hold:\begin{itemize}
\item $G_1$ is a subgraph of $G_2$.
\item There are embeddings $L\xrightarrow{f}G_1$ and $R\xrightarrow{g}G_2$ such that $f(v)=g(v)$ for all vertices in $L$.
\item If there are embeddings $G_1\xrightarrow{f'}G$ and $R\xrightarrow{g'}G$ such that $f'(f(v)) = g'(v)$ for all vertices in $L$, then there is an embedding $G_2\xrightarrow{h}G$ such that $f'(v) = h(v)$ for all vertices in $G_1$.
\end{itemize}
\end{definition}
Like in our variation of \gr{}s, we use a \wps{} as a variation of the categorical pushout\footnote{Pushouts in a category with embeddings as arrows} that is typically used in graph rewriting, to ensure that chains are formed.
In such a (weak) pushout, the requirement of subgraphs is missing, making the entire definition symmetrical ($G_1$ and $R$ can be switched).
A pushout, as compared to a weak pushout, additionally requires that the embeddings $f'$ and $g'$ at the end of our definition, is unique.
These subtle differences arise out of our need to form chains, which aren't typical structures in graph rewriting.

\begin{definition}[(Simple) \WPC{}]\label{def:pushoutchain}
Let $S$ be a chain with $S(i) = (\mathcal{L},E_i,V_i)$, and let $\mathcal{R}$ be a set of \gr{}s.
If for each $i$, either $S(i) = S(i+1)$ or there exists an $r\in\mathcal{R}$ such that $(S(i),S(i+1))$ is a \wps{} for $r$, then $S$ is a \emph{simple \wpc{}} under $\mathcal{R}$.
\emph{\Wpc{}s} are inductively defined:
\begin{enumerate}
\item every simple \wpc{} under $\mathcal{R}$ is a \wpc{} under $\mathcal{R}$,
\item if for each $i$, there exists an $s$, which is a \wpc{} under $\mathcal{R}$ with $s(0) = S(i)$ and $s(\infty) = S(i+1)$, then $S$ is a \wpc{} under $\mathcal{R}$, \label{enum:inductive}
\item nothing else is a \wpc{}.
\end{enumerate}
\end{definition}
For most of this paper, it suffices to consider simple \wpc{}s.

There is a way to draw \wps{}s that is convenient in practice, although it can leave parts implicit.
On a \wps{} $(G_1,G_2)$ for $(L,R)$, as drawn in Figure~\ref{fig:wps}, large vertices indicate vertices in the image of $f$ for $L\xrightarrow{f}G_1$.
The applied \gr{} is that of Figure~\ref{fig:rulcompact}.
Edges in $\hat{f}(L)$ are drawn slightly thicker.
The corresponding $\hat{g}$ for $R\xrightarrow{g}G_2$ is drawn with dotted lines.
Since the drawing is of a \wps{}, small vertices connected to dotted lines are in $G_2$ but not in $G_1$.
The graph $G_1$ is the graph of Figure~\ref{fig:graph}, and $G_2$ is the graph in Figure~\ref{fig:graph2}.

\begin{figure}[tbp]
\centering
\begin{subfigure}[b]{0.4\textwidth}\centering
\begin{tikzpicture}[-latex,thin]
\node[vertex,label=above:$a$](1) {}; 
\node[vertex,label=above:$b$,inner sep=2pt](2) [right =of 1] {};
\node[vertex,label=below:$c$,inner sep=2pt](3) [below =of 2] {};
\path (1) edge[bend left] node[label=above:$k$] {} (2);
\path (3) edge[bend left] node[label=below:$m$] {} (1);
\path (2) edge[thick] node[label=right:$l$] {} (3);
\path (1) edge [loop left,in=135,out=-135,distance=1cm] node[label=$l$] {} (1);
\node[vertex,label=below:$d$](x) [right=of 2] {};
\path (2) edge[dotted] node[label=below:$l$] {} (x);
\path (x) edge[dotted] node[label=below:$l$] {} (3);
\end{tikzpicture}
\caption{A \wps{}, compactly}\label{fig:wps}
\end{subfigure}
\begin{subfigure}[b]{0.4\textwidth}\centering
\begin{tikzpicture}[-latex,semithick]
\node[vertex,label=above:$a$](1) {}; 
\node[vertex,label=above:$b$](2) [right =of 1] {};
\node[vertex,label=below:$c$](3) [below =of 2] {};
\path (1) edge[bend left] node[label=above:$k$] {} (2);
\path (3) edge[bend left] node[label=below:$m$] {} (1);
\path (2) edge node[label=right:$l$] {} (3);
\path (1) edge [loop left,in=135,out=-135,distance=1cm] node[label=$l$] {} (1);
\node[vertex,label=below:$d$](x) [right=of 2] {};
\path (2) edge node[label=below:$l$] {} (x);
\path (x) edge node[label=below:$l$] {} (3);
\end{tikzpicture}
\caption{Result of the step}\label{fig:graph2}
\end{subfigure}
\caption{On \wps{}s}
\end{figure}

A \wpc{} does not necessarily have a consequence graph as its supremum: we can construct a \wpc{} with $S(i) = G$ for any graph $G$.
However, the following holds:
\begin{lemma}\label{lem:pushoutchain}
If $S$ is a \wpc{} under $\mathcal{R}$ and $S(\infty)$ is a consequence graph maintaining $\mathcal{R}$, then $S(\infty)$ is a least consequence graph of $S(0)$ maintaining $\mathcal{R}$.
\end{lemma}
\begin{proof}
Let $G$ be a consequence graph.
We first consider the case in which $S$ is a simple \wpc{}.
By induction on $i$, we prove that $(S(0),S(i))$ is maintained in $G$:
For $i=0$, $(S(0),S(0))$ is trivially maintained in any graph.
For $S(i+1)$, assume $(S(0),S(i))$ is maintained in $G$ by induction.
If $S(i+1) = S(i)$, then $S(i+1)$ is trivially embedded in $G$.
If $S(i)\neq S(i+1)$, then $(S(i),S(i+1))$ is a \wps{} for some $(L,R) \in \mathcal{R}$.
Given an embedding $S(0)\rightarrow G$, as $L$ is embedded in $S(i)$, transitively $L$ is embedded in $G$.
Since $G$ is a consequence graph, $R$ is embedded in $G$ such that, by definition, there exists an embedding of $S(i+1)$ into $G$.
We conclude that for all $i$, $S(i)$ is embedded in $G$.
To conclude that $S(\infty)$ is also embedded in $G$, note that the individual embeddings of $S(i)$ in $G$ have a limit (each embedding function is contained in the next by $f'(v) = h(v)$).
The case in which the \wpc{} $S$ is not simple follows inductively from composing embeddings.
Therefore $S(\infty)$ is a least consequence graph. \qed
\end{proof}

A chain that is both fair and a \wpc{} is called a fair \wpc{}.
A fair \wpc{} has a least consequence graph as its supremum.
This gives a way to create least consequence graphs, which we'll come back to in Section~\ref{sec:procedure}.

\section{Translation between \TR{}s and \GR{}s}\label{sec:Rules}
This section shows how to turn \tr{}s into \gr{}s.
For every \tr{}, there is a corresponding \gr{} that is maintained if and only if the \tr{} holds.
This allows us to use \gr{}s in order to reason about \tr{}s.
We introduce a translate function $\transL : \exprL \rightarrow \mathbb{G}_\mathcal{L}$ in Definition~\ref{def:translation} to make precise which graph belongs to a term.
Lemma~\ref{lem:translatesound} states how the two correspond in the case of \tr{}s of the shape $\_ \sqsubseteq \_$.
Using Lemma~\ref{lem:subseteq}, this means we can encode a set of \tr{}s as a set of \gr{}s.

\begin{definition}[Translation]\label{def:translation}
Given a term $\mathbb{e}$, we say that $\transL(\mathbb{e})$ is the \emph{translation} of $\mathbb{e}$.
We define $\transL : \exprL \rightarrow \mathbb{G}_\mathcal{L}$ as follows:
\begin{align*}
\transL(l)&=\left(\mathcal{L},\{0,1\},\{(l,0,1)\}\right)\\
\transL(\mathbb{e} \converse{})&=\hat{f}(\transL(\mathbb{e})) \text{\quad with $f(v) = 1 - v$ for $v<2$ and $f(v) = v$ for $v\geq 2$.}\\
\transL(\mathbb{e}_1 \fcmp \mathbb{e}_2)&=\hat{f_1}(\transL(\mathbb{e}_1)) \cup \hat{f_2}(\transL(\mathbb{e}_2))\\
&\qquad\text{with $f_1(0) = 0$ and $f_1(v) = v + \left|\transL(\mathbb{e}_2)\right| - 1$ for $v\neq 0$,}\\
&\qquad\text{and $f_2(0) = \left|\transL(\mathbb{e}_2)\right|$ and $f_2(v) = v$ for $v\neq 0$.}\\
\transL(\mathbb{e}_1 \sqcap \mathbb{e}_2)&=\transL(\mathbb{e}_1) \cup \hat{f}(\transL(\mathbb{e}_2))\\
&\qquad\text{with $f(v) = v$ for $v<2$ and $f(v) = v + \left|\transL(\mathbb{e}_1)\right| - 2$ for $v\geq 2$.}
\end{align*}
For notational convenience, $\transL\left(\mathbb{e}_L \sqsubseteq \mathbb{e}_R\right) = \left(\transL(\mathbb{e}_L),\transL(\mathbb{e}_L \sqcap \mathbb{e}_R)\right)$.
It follows that $\transL\left(\mathbb{e}_L \sqsubseteq \mathbb{e}_R\right)$ is a \gr{}.
\end{definition}

As an example of how the translation works, the graphs in Figure~\ref{fig:rulpair} are $\transL(l)$ and $\transL\left(l\sqcap l\fcmp l\right)$ respectively.
As a whole, the \gr{} in Figure~\ref{fig:rulpair} is $\transL\left(l\sqsubseteq l\fcmp l\right)$.

The vertices $0$ and $1$ of $\transL(\mathbb{e})$ can intuitively be understood as the variables $x$ and $y$ as in Definition~\ref{def:semantics}.
Lemma~\ref{lem:translate} makes this precise:
\begin{lemma}\label{lem:translate}
$(v_0,v_1) \in \sem{\mathbb{e}}G$ if and only if there is an $f$ such that $\transL(\mathbb{e})\xrightarrow{f}G$ with $f(i) = v_i$ for $i<2$.
\end{lemma}
\begin{proof}
The statement follows by induction on $\mathbb{e}$, using that the vertices in $\transL(\mathbb{e})$ are $\left\{i \mid i \in \mathbb{N} \wedge i<\left|\transL(\mathbb{e})\right|\right\}$.\qed
\end{proof}

We can use Lemma~\ref{lem:translate} to show a connection between \gr{}s and \tr{}s:
\begin{lemma}\label{lem:translatesound}
A \tr{} $\mathbb{e}_L \sqsubseteq \mathbb{e}_R$ holds in $G$ if and only if $\left(\transL(\mathbb{e}_L), \transL(\mathbb{e}_L \sqcap \mathbb{e}_R)\right)$ is maintained in $G$.
\end{lemma}
\begin{proof}
Suppose the \tr{} holds in $G$, and $\transL(\mathbb{e}_L)\xrightarrow{f}G$.
It follows from Lemma~\ref{lem:translate} that $(f(0),f(1)) \in \sem{\mathbb{e}_L}G$.
As the \tr{} holds, $(f(0),f(1)) \in \sem{\mathbb{e}_R}G$.
Using Lemma~\ref{lem:translate}, take $g$ with $\transL(\mathbb{e}_R)\xrightarrow{g}G$ and $f(v) = g(v)$ for $v < 2$.
Following Definition~\ref{def:translation}, construct $g'$ such that $\transL(\mathbb{e}_L \sqcap \mathbb{e}_R)\xrightarrow{g'}G$ and $f(v) = g'(v)$ for $v$ in the vertices of $\transL(\mathbb{e}_L)$.

For the other direction, suppose $\left(\transL(\mathbb{e}_L),\transL(\mathbb{e}_L \sqcap \mathbb{e}_R)\right)$ is maintained in $G$, and let $(x,y) \in \sem{\mathbb{e}_L}G$.
By Lemma~\ref{lem:translate}, there is an $f$ such that $\transL(\mathbb{e}_L)\xrightarrow{f}G$ with $f(0) = x$ and $f(1) = y$.
Since the \gr{} is maintained, there is a $g$ such that $\transL(\mathbb{e}_L \sqcap \mathbb{e}_R)\xrightarrow{g}G$ with $g(0)=f(0)=x$ and $g(1)=f(1)=y$.
Again using Lemma~\ref{lem:translate}, $(x,y)\in\sem{\mathbb{e}_L \sqcap \mathbb{e}_R}G = \sem{\mathbb{e}_L}G \cap \sem{\mathbb{e}_R}G \subseteq \sem{\mathbb{e}_R}G$, so the \tr{} holds in $G$. \qed
\end{proof}

We use \gr{}s to deal with the requirements in Definition~\ref{def:standard} to make a graph standard, in a way similar to Lemma~\ref{lem:translatesound}.
We give a set of \gr{}s that make checking if a standard graph exists easy:
A standard graph exists provided that $\sem{\bot}G = \{\}$, and that a set of additional \gr{}s, which we will call the standard-rules, is maintained.
This motivates the following definitions:

\begin{definition}[Conflict (Free)]\label{def:conflict}
Let $\bot \in \mathcal{L}$.
The relation symbol $\bot$ stands for an empty relation.
A graph for which $\sem{\bot}G = \{\}$ is \emph{conflict free}.
If $G = (\mathcal{L},V,E)$ is conflict free, we have $\forall x y. (\bot,x,y) \not\in E$,
so we call any edge $(\bot,x,y)$ a \emph{conflict}.
\end{definition}
\begin{definition}[Top-rule]
Let $\top \in \mathcal{L}$.
The relation symbol $\top$ stands for the full relation.
We refer to the \gr{} $((\mathcal{L},\{0,1\},\{\}),(\mathcal{L},\{0,1\},\{(\top,0,1)\}))$ as the \emph{top-rule}, since any graph $G = (\mathcal{L},V,E)$ satisfies $\sem{\top}G = \{(x,y) \mid x,y \in V\}$ if and only if $G$ maintains the top-rule.
\end{definition}
\begin{definition}[Nonempty-rule]
Let $\bot,\top \in \mathcal{L}$.
The \gr{} $((\mathcal{L},\{\},\{\}),(\mathcal{L},\{0\},\{\}))$ is called the \emph{nonempty-rule}.
A graph $G = (\mathcal{L},V,E)$ maintains the nonempty-rule if and only if $V \neq \{\}$.
\end{definition}

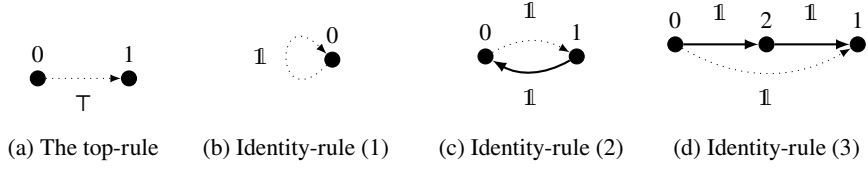
\begin{figure}[tbp]
\centering
\begin{subfigure}[b]{0.2\textwidth}\centering
\begin{tikzpicture}[-latex,thin]
\node[vertex,label=above:$0$,inner sep=2pt](0) {}; 
\node[vertex,label=above:$1$,inner sep=2pt](1) [right =of 0]{}; 
\path (0) edge[dotted] node[label=below:$\top$] {} (1);
\end{tikzpicture}
\caption{The top-rule}\label{fig:toprule}
\end{subfigure}
\begin{subfigure}[b]{0.25\textwidth}\centering
\begin{tikzpicture}[-latex,thin]
\node[vertex,label=above:$0$,inner sep=2pt](0) {};
\path (0) edge[dotted,in=135,out=-135,loop,distance=1cm] node[label=left :$\ident{}$] {} (0);
\end{tikzpicture}
\caption{Identity-rule~(\ref{idt:reflexive})}\label{fig:reflrule}
\end{subfigure}
\begin{subfigure}[b]{0.25\textwidth}\centering
\begin{tikzpicture}[-latex,thin]
\node[vertex,label=above:$0$,inner sep=2pt](0) {}; 
\node[vertex,label=above:$1$,inner sep=2pt](1) [right =of 0]{}; 
\path (1) edge[thick,bend left] node[label=below :$\ident$] {} (0);
\path (0) edge[dotted,bend left] node[label=above :$\ident$] {} (1);
\end{tikzpicture}
\caption{Identity-rule~(\ref{idt:symmetric})}\label{fig:symrule}
\end{subfigure}
\begin{subfigure}[b]{0.25\textwidth}\centering
\begin{tikzpicture}[-latex,thin]
\node[vertex,label=above:$0$,inner sep=2pt](0) {}; 
\node[vertex,label=above:$2$,inner sep=2pt](2) [right =of 0]{}; 
\node[vertex,label=above:$1$,inner sep=2pt](1) [right =of 2]{}; 
\path (0) edge[thick] node[label=above :$\ident$] {} (2);
\path (2) edge[thick] node[label=above :$\ident$] {} (1);
\path (0) edge[dotted,bend right] node[label=below :$\ident$] {} (1);
\end{tikzpicture}
\caption{Identity-rule~(\ref{idt:transitive})}\label{fig:nonemptyrule}
\end{subfigure}
\caption{Several standard-rules}
\end{figure}

A conflict-free graph $G$ that maintains the top-rule, satisfies $\sem{\top}G \neq \sem{\bot}G$ if and only if it maintains the nonempty-rule.

The relation symbol $\ident$ models the identity relation $\{(x,x) \mid x \in V\}$.
However, we do not let $\sem{\ident}G$ represent this relation directly.
Instead, we let $\ident$ stand for an equivalence relation and ensure that we can make a graph based on equivalence classes, in which $\sem{\ident}G = \{(x,x) \mid x \in V\}$ holds.
\begin{definition}[Identity-rules]
Given a set of relation symbols $\mathcal{L}$, we say that the following set of \gr{}s are the \emph{identity-rules} for $\mathcal{L}$:
\begin{align}
\big((\mathcal{L},\{0\},\{\})&\ ,\ \ (\mathcal{L},\{0\},\{(\ident,0,0)\})\big)\label{idt:reflexive} \\
\transL\big(\ident \converse{} &\sqsubseteq \ident \big)\label{idt:symmetric} \\
\transL\big(\ident \fcmp \ident &\sqsubseteq \ident \big)\label{idt:transitive} \\
\forall l\in\mathcal{L}.\ \transL\big(\ident \fcmp l \fcmp \ident &\sqsubseteq l \big)\label{idt:eachl}
\end{align}
\end{definition}

Identity-rules~(\ref{idt:reflexive}) to~(\ref{idt:eachl}) can be understood as ensuring $\ident$ is reflexive, symmetric, transitive, and a congruence respectively.
The identity-rules hold under the standard semantics of $\ident$, that is: if for some graph $G = (\mathcal{L},E,V)$, we have $\sem{\ident}G = \{(x,x) \mid x \in V\}$ then the identity-rules are maintained in $G$.
The following lemma speaks about the other direction:
\begin{lemma}\label{lem:identrules}
Let $G = (\mathcal{L},V,E)$ be a graph in which the identity-rules for $\mathcal{L}$ are maintained.
There is an idempotent $f$ such that $\hat{f}(G) \xrightarrow{\subseteq} G$, and $\sem{\ident}{\hat{f}(G)} = \{(f(x),f(x)) \mid x \in V\}$.
\end{lemma}
\begin{proof}
Since the first three identity-rules for $\mathcal{L}$ are maintained in $G$, $\sem{\ident}G$ is an equivalence relation on $V$.
Let $f$ be some function that takes a canonical element from the equivalence class.
It follows that $\sem{\ident}{\hat{f}(G)} = \{(f(x),f(x)) \mid x \in V\}$, and it remains to be shown that $\hat{f}(G) \xrightarrow{\subseteq} G$.
For the vertices, this is immediate.
For the edges: For all $(l,x,y) \in E$ we show that $(l,f(x),f(y))\in E$.
By our choice of $f$, $(\ident,f(x),x)\in E$ and $(\ident,y,f(y))\in E$.
Suppose $(l,x,y)\in E$.
Since Identity-rule~(\ref{idt:eachl}) for $l$ is maintained in $G$, we get $(l,f(x),f(y)) \in E$.
Therefore $\hat{f}(G) \xrightarrow{\subseteq} G$. \qed
\end{proof}

Lemma~\ref{lem:identrules} gives us exactly the desired semantics for $\ident$: for $(\mathcal{L},V',E') = \hat{f}(G)$, we have $\sem{\ident}{\hat{f}(G)} = \{(x,x) \mid x \in V'\}$.
Furthermore, it states that $\hat{f}(G)$ and $G$ are mutually embedded ($G \rightarrow \hat{f}(G)$ holds for all $f$).

We now proceed to introduce constants, through a set of \tr{}s.
This characterisation is similar to how points are characterized in relation algebra, see for instance work by Schmidt and Str\"ohlein~\cite{Schmidt:1985aa}.
If $c$ is a constant, then $p = c\fcmp \top$ is a point (sometimes called a right ideal).
The corresponding constant can be retrieved from a point: $c = p \fcmp p\converse{}$.
Our presentation here in terms of constants rather than points is a matter of personal preference.
These rules state that the relation $\sem{c}G$ should be nonempty, the cross-product of two sets, and a subset of the identity relation.
Finally, we state that for two different constants, $\sem{c_1}G$ and $\sem{c_2}G$ should be non-overlapping.
\begin{definition}[Constant-rules, Standard-rules]
Let $\mathcal{L}$ be a standard set of labels with constants $\mathcal{C}$, we say that the following set of \gr{}s are the \emph{constant-rules} for $\mathcal{C}$:
\begin{align}
\forall c\in\mathcal{C}.&& \transL\big(\top &\sqsubseteq \top \fcmp c \fcmp \top \big)\label{pr:exists} \\
\forall c\in\mathcal{C}.&& \transL\big(c \fcmp \top \fcmp c &\sqsubseteq c \big) \label{pr:crossprod} \\
\forall c\in\mathcal{C}.&& \transL\big(c &\sqsubseteq \ident \big)\label{pr:ident} \\
\forall c_1, c_2\in\mathcal{C}.\  c_1 \neq c_2 \Rightarrow && \transL\big(c_1 \fcmp c_2 &\sqsubseteq \bot \big)\label{pr:neq}
\end{align}
The top-, nonempty-, identity-, and constant-rules together are called the \emph{standard-rules} for $\mathcal{C}$ and $\mathcal{L}$, written $\mathcal{S}_{\mathcal{C},\mathcal{L}}$.
\end{definition}

Similar to our treatment of $\ident$, we would like to find an $f$ such that $\sem{c}{\hat{f}(G)} = \{(c,c)\}$.
The $f$ of Lemma~\ref{lem:identrules} gives us a graph that is isomorphic to one in which $\forall c\in\mathcal{C}.\ \sem{c}{\hat{f}(G)} = \{(c,c)\}$ holds, provided that $G$ is conflict free and maintains the standard-rules.
Lemma~\ref{lem:existsgraph} says that, for finding a model with `standard semantics', it suffices to find a conflict free graph that maintains the standard-rules.

\begin{lemma}\label{lem:existsgraph}
Let $\mathcal{L}$ be a standard set of labels with constants $\mathcal{C}$.
Let $\mathcal{T}$ be a set of \tr{}s over $\mathcal{L}$ of the shape $\_ \sqsubseteq \_$.
We define $\mathcal{R} = \mathcal{S}_{\mathcal{C},\mathcal{L}} \cup \{\transL(t)\mid t\in \mathcal{T}\}$.
Let $G'$ be a conflict free consequence graph maintaining $\mathcal{R}$, then there is a graph $G = (\mathcal{L},E,V)$, and functions $f$ and $g$ such that:
\begin{enumerate}
\item $G = \hat{f}(\hat{g}(G)) =\hat{f}(G')$, and $\hat{g}(G) \xrightarrow{\subseteq} G'$.
\item The graph $G$ is standard.
\item Every \tr{} in $\mathcal{T}$ holds in $G$.
\end{enumerate}
\end{lemma}
\begin{proof}
We begin the proof by constructing $G$ and $f$, based on $G' = (\mathcal{L},E',V')$.
By Lemma~\ref{lem:identrules}, there is an idempotent function $h$ with  $\sem{\ident}{\hat{h}(G')} = \{(h(x),h(x)) \mid x \in V'\}$.
Top- and nonempty-rules are maintained in $G'$, so by constant-rule~(\ref{pr:exists}), there are vertices $v_1,v_2\in V'$ with $(c,v_1,v_2)\in E'$ for each $c\in\mathcal{C}$.
Let $m : \mathcal{C} \to V'$ such that for each $c$, $\exists v\in V'.\ (c,m(c),v)\in E'$, therefore $\exists v\in V'.\ (c,h(m(c)),h(v))\in E'$.
Using constant-rule~(\ref{pr:ident}), it follows that $h(m(c)) = h(v)$, so $(c,h(m(c)),h(m(c))) \in E'$.
From constant-rule~(\ref{pr:neq}) and that $G'$ is conflict free, we get $(c_1,h(m(c_2)),h(m(c_2)))\not\in E'$ iff $c_1\neq c_2$.
We conclude that $h \circ m$, the function that maps $c\in\mathcal{C}$ to $h(m(c))$, is injective, so $m$ is injective.
Therefore, there is a $V$ and an $m'$ with $\mathcal{C}\subseteq V$ such that $m' : V \to V'$ is bijective and $m(c) = m'(c)$ for $c\in \mathcal{C}$.
Let $f = h \circ m'$, defining $G = \hat{f}(G')$.
Let $g$ be the inverse of $m'$, giving $\hat{f}(\hat{g}(G)) = \hat{h}(G) = G$ since $h$ is idempotent.
We have $\sem{\ident}G = \{(x,x) \mid x\in V\}$ by our choice of $h$.
Also $G$ is a consequence graph of $\mathcal{R}$ since $G = \hat{f}(G')$ and $G'$ is a consequence graph of $\mathcal{R}$.
From $(c,h(m(c)),h(m(c))) \in E'$ we get $(c,c)\in\sem{c}G$.
Using constant-rule~(\ref{pr:crossprod}) and constant-rule~(\ref{pr:ident}), now with $\sem{\ident}G = \{(x,x) \mid x\in V\}$, we get $\{(c,c)\}=\sem{c}G$.
All properties now follow.\qed
\end{proof}

\section{A Procedure to Find a Standard Graph}\label{sec:procedure}
A set of \tr{}s is satisfiable if and only if there is no $i$ such that $S(i)$ contains a conflict in a corresponding fair \wpc{}.
This follows from the previous sections as follows:
Given a set of \tr{}s $\mathcal{T}'$ with relation symbols $\mathcal{L}$, Lemma~\ref{lem:subseteq} shows that we can find an equivalent set of \tr{}s $\mathcal{T}$ such that each \tr{} is of the shape $\_ \sqsubseteq \_$.
We derive a set of \gr{}s $\mathcal{R}$ that includes the standard-rules and the translation of the \tr{}s in $\mathcal{T}$.
By making a fair \wpc{} $S$ starting in the empty graph, we obtain a supremum that is a least consequence graph of $\mathbb{0}_\mathcal{L}$ maintaining $\mathcal{R}$.
If this graph contains a conflict, then any graph maintaining $\mathcal{R}$ will, so $\mathcal{T}'$ is unsatisfiable.
If not, we can apply Lemma~\ref{lem:existsgraph} to find a model for $\mathcal{T}'$.
In this section, we look at constructing fair \wpc{}s, based on a set of \gr{}s $\mathcal{R}$ that include the standard-rules.

\subsection{An Algorithm for Fair \WPC{}s}\label{sec:algorithm}
Assume that the set of \tr{}s $\mathcal{T}$ is finite.
Consequently, only finitely many relation symbols $\mathcal{L}$ are used in those \tr{}s.
We restrict $\mathcal{L}$ to those relation symbols that are actually used in $\mathcal{T}$.
This makes the corresponding set of \gr{}s $\mathcal{R}$ (including the standard-rules) finite.
Thus, we can construct a fair \wpc{} for $\mathcal{R}$.
Algorithm~\ref{alg:fairpushoutchain} gives a procedure for this.

\begin{algorithm}
\SetKwInOut{Input}{Input}
\SetKwInOut{Effect}{Effect}
\underline{function ProduceChain} $(n \in \mathbb{N},\mathcal{L},E,\mathcal{R})$\;
\Input{A set of edges $E$ such that $G = (\mathcal{L},\{i\mid i\in\mathbb{N},\ i < n\},E)$ is a graph. A finite set of finite \gr{}s $\mathcal{R}$ with relation symbols $\mathcal{L}$.}
\Effect{Produces an infinite list of graphs that are a fair \wpc{} starting in $G$.}
Let $G = (\mathcal{L},\{i\mid i\in\mathbb{N},\ i < n\},E)$, produce $G$\;
Let $W = \{\}$ be our worklist\;
\For{$(L,R)\in\mathcal{R}$}{
  Take $V$ such that $(\_,V,\_)\in L$\;
  \For{$f$ such that $L\xrightarrow{f}G$}{
    \If{There is no $g$ such that $R\xrightarrow{g}G$ with $\forall v\in V .\ f(v) = g(v)$}{
      Let $N\in\mathbb{N}$ be the maximum of $f(v)$\;
      Add $(N,L,R,f)$ to $W$\;
    }
  }
}
\eIf{$W$ is empty}{
  ProduceChain$(n,\mathcal{L},E,\{\})$\label{lnr:stop}\;
}{
  Take $(N,L,R,f)\in W$ such that $N$ is minimal\label{lnr:pickGraph}\;
  Take $V, V'$ such that $(\_,V,\_)= L$ and $(\_,V',\_)= R$\;
  Let $\Delta = V' - V$\;
  Let $V'' = \{i\mid  i\in\mathbb{N},\ i<n+|\Delta|\}$\;
  Take $g : V' \rightarrow V''$ such that $g(v) = f(v)$ for $v\in V$
   and $g(\delta) \geq n$ for $\delta\in\Delta$ such that $g(\delta_1) \neq g(\delta_2)$ if $\delta_1\neq\delta_2$ for $\delta_1,\delta_2 \in \Delta$\;
  Take $E'$ such that $(\mathcal{L},V'',E') = G \cup g(R)$\label{lnr:makeGraph}\;
  ProduceChain$(n+|\Delta|,\mathcal{L},E',\mathcal{R})$\;
}
\caption{Construct a fair \wpc{} starting in its input}\label{alg:fairpushoutchain}
\end{algorithm}

\begin{lemma}\label{lem:algSound}
Algorithm~\ref{alg:fairpushoutchain} constructs a fair \wpc{} starting in $G$ under $\mathcal{R}$, the limit of which is a least consequence graph of $G$ under $\mathcal{R}$.
\end{lemma}
\begin{proof}
The algorithm constructs a \wpc{}, because the graph constructed on Line~\ref{lnr:makeGraph} is part of a \wps{} for a \gr{} in $\mathcal{R}$.
Let $S : \mathbb{N} \to \mathbb{G}_\mathcal{L}$ describe the \wpc{} generated (with $S(0) = G$).
Pick an arbitrary $N$.
Since the set of \gr{}s is finite, also the number of functions $f$ with $f(v) \leq N$ that embed left-hand sides of \gr{}s into $S(\infty)$ is finite.
For some $i$, all such embeddings are in $S(i)$.
If an embedding is picked on Line~\ref{lnr:pickGraph}, there is a $g$ such that $R\xrightarrow{g}S(\infty)$, since such a $g$ is added on Line~\ref{lnr:makeGraph}.
Therefore, for each embedding $f$ with $f(v) \leq N$ such that $L\xrightarrow{f}G$, there is a $g$ such that $R\xrightarrow{g}S(\infty)$ with $f(v)=g(v)$.
The domain for every such $f$ is finite, so we can pick an $N$ for every $f$ such that $f(v) \leq N$.
Therefore, the \wpc{} is fair.
Lemma~\ref{lem:fairchain} and~\ref{lem:pushoutchain} complete the proof.
\qed
\end{proof}

The algorithm can be changed into a semi-decision procedure to decide whether the limit contains a conflict:
If $G$ contains a conflict, then any limit in which $G$ occurs will contain the conflict.
Therefore, if we are only interested in whether the limit has a conflict, we can abort the algorithm as soon as $G \cup g(R)$ in Line~\ref{lnr:makeGraph} has a conflict.
Vice versa, if the limit has a conflict, then there will be a graph $G$ in some iteration of the algorithm that has that conflict.
This gives a semi-decision procedure.
We can use this to decide consistency, using $\mathbb{0}_\mathcal{L}$ as the initial graph.

The same procedure can be used to prove entailment.
Say we wish to determine if $\mathcal{T}$ entails $\phi$ for a problem on a standard set of labels $\mathcal{L}$, for $\phi$ equal to $\mathbb{e}_1 \sqsubseteq \mathbb{e}_2$.
Assume without loss of generality that $l\not\in\mathcal{L}$.
We introduce a new label $l$: $\mathcal{L}' = \mathcal{L} \cup \{l\}$.
Let $\mathcal{T}' = \mathcal{T} \cup \{l \sqsubseteq \mathbb{e}_1, \mathbb{e}_2 \sqcap l \sqsubseteq \bot\}$.
Let $\mathcal{R}$ be the standard rules plus the derived rules of $\mathcal{T}'$.
This time, run the algorithm with $(\mathcal{L}',\{0,1\},\{(l,0,1)\})$ as the initial graph: we obtain a least consequence graph maintaining $\mathcal{R}$.
If this graph does not contain a conflict, there is a standard graph $G$ in which $\mathcal{R}$ is maintained, and therefore $\mathcal{T}$ holds in $G$, but $\phi$ does not hold as $\sem{l}G \subseteq \sem{\mathbb{e}_1}G$ but $\sem{l}G \cap \sem{\mathbb{e}_1}G = \{\}$ for $\sem{l}G$ nonempty, since $(\mathcal{L}',\{0,1\},\{(l,0,1)\}){}\ \xrightarrow{} G$.
If the obtained graph does contain a conflict, then all consequence graphs of $\mathcal{R}$ with nonempty $l$ contain a conflict.
Suppose $G$ is standard, each of $\mathcal{T}$ holds, there is a pair (labeled $l$) in $\sem{\mathbb{e}_1}G$, and that pair is not in $\sem{\mathbb{e}_2}G$, then we get a contradiction to the statement that all consequence graphs of $\mathcal{R}$ with nonempty $l$ contain a conflict.
In other words: for each standard $G \vDash \mathcal{T}$, we have $\sem{\mathbb{e}_1}G \subseteq \sem{\mathbb{e}_2}G$, so $\mathcal{T}$ entails $\phi$.
This shows that a least consequence graph can be used to decide entailment.
By terminating our procedure when a conflict is found, we can prove entailment if it holds (and do not terminate otherwise).
This can be extended to $\phi$ of the shape $\mathbb{e}_1 = \mathbb{e}_2$, by applying this procedure to both $\mathbb{e}_1 \sqsubseteq \mathbb{e}_2$ and $\mathbb{e}_2 \sqsubseteq \mathbb{e}_1$.

There is another case in which we can abort: once the graph maintains all \gr{}s in $\mathcal{R}$, we hit Line~\ref{lnr:stop}, and $G$ is equal to the limit.
In such a case, we have found the limit of the chain given by Algorithm~\ref{alg:fairpushoutchain}, and can immediately decide whether or not it is conflict free.
Unfortunately, even if conflict free graphs that maintain all \gr{}s exist (so by definition of least consequence graph, the limit is conflict free), we do not necessarily hit this case.
Section~\ref{sec:undecidable} shows that we cannot hope to find an algorithm that decides whether or not a conflict free consequence graph exists.

\subsection{Optimizations for Implementations}
We discuss some possible optimizations for the purpose of showing correctness of the algorithm described by the author in an earlier paper~\cite{amperspiegelRAMICS}.
The earlier algorithm is not Algorithm~\ref{alg:fairpushoutchain}, but an optimized version thereof.
We only describe a few optimizations, that suffice to show that the algorithm presented earlier is correct as well.

As optimizations, we allow changing the outcome of the algorithm, but require that the proof of Lemma~\ref{lem:algSound} remains valid.
In particular, instead of the graph $G \cup g(R)$ constructed on Line~\ref{lnr:makeGraph}, we can make a larger graph $S(\infty)$ if $G \cup g(R) \subseteq S(\infty)$ and $S(\infty)$ is the limit of a (not necessarily fair) \wpc{}.
Through this change, the algorithm no longer constructs simple \wpc{}s, but Lemma~\ref{lem:algSound} still holds.

As an instance of this, observe that we can combine \gr{}s, as this is a form of combining \wps{}s:
suppose $(L,R)$ and $(L',R')$ are \gr{}s in $\mathcal{R}$, such that $L' \xrightarrow{f} R$.
Then we can find an $R''$ such that $(R,R'')$ is a \wps{} of $(L',R')$.
We can then safely replace the \gr{} $(L,R)$ for $(L,R'')$ in $\mathcal{R}$, as a \wps{} of $(L,R'')$ is the limit of a chain that satisfies the aforementioned condition.

Apart from changing the set of \gr{}s $\mathcal{R}$, we can change the algorithm such that the standard-rules are always maintained after each step.
Let $G'$ be a graph constructed in that way.
According to Lemma~\ref{lem:existsgraph}, we represent the graph $G'$ by the graph $\hat{f}(G')$, making it such that we do not need to store the relation-symbols $\top$, $\bot$, $\ident{}$, or the constants in $\mathcal{C}$.
We do need to keep track of which vertices originally belong to which equivalence classes, in order to be able to produce the underlying $G'$ in each step.
Since the function $f$ possibly maps several vertices of $G'$ to one vertex in $\hat{f}(G')$, the original graph $\hat{f}_{i}(G'_{i})$ is not necessarily a subgraph of the newly generated graph $\hat{f}_{i+1}(G'_{i+1})$.
On the other hand, if we are only interested in whether or not there is a conflict in the least consequence graph, then we only need to keep track of the least vertex of each class such that the $N$ chosen on Line~\ref{lnr:pickGraph} corresponds to a minimal embedding of $f$.
This is precisely the algorithm proposed in the earlier paper~\cite{amperspiegelRAMICS}, showing it is a semi-decision procedure for deciding whether a least consequence graph contains a conflict.

\subsection{Example Run of the Optimized Algorithm}
We return to one of the examples given in Section~\ref{sec:background}: the entailment problem that asks whether $\mathcal{T} = \{{\tt r} = {\tt i}\fcmp{\tt i}\converse{}\ ,\ {\tt i}\converse{} \fcmp {\tt i} \sqsubseteq \ident{}\}$ entails ${\tt r}\fcmp{\tt r} \sqsubseteq {\tt r}$.
We construct a $\mathcal{T}'$ for the entailment problem as described in Section~\ref{sec:algorithm}:
\begin{align*}
\mathcal{T}' = \{\ & {\tt r} \sqsubseteq {\tt i}\fcmp{\tt i}\converse{}\ ,\ {\tt i}\fcmp{\tt i}\converse{} \sqsubseteq {\tt r} \ ,\  {\tt i}\converse{} \fcmp {\tt i} \sqsubseteq \ident{} \ ,\  l \sqsubseteq {\tt r}\fcmp{\tt r} \ ,\ l \sqcap {\tt r} \sqsubseteq \bot\ \}
\end{align*}
Using the translation of Section~\ref{sec:Rules}, Figure~\ref{fig:exgr} gives the \gr{}s we work with.
We use the optimizations just described, and do not restate the standard-rules.

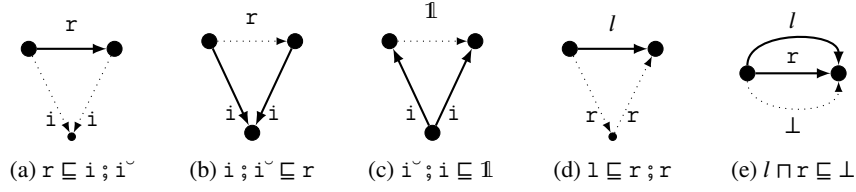
\begin{figure}[htbp]
\centering
\begin{subfigure}[b]{0.19\textwidth}\centering
\begin{tikzpicture}[-latex,thin]
\node[](c) {};
\node[vertex,inner sep=2pt](0) [left = 10pt of c] {};
\node[vertex,inner sep=2pt](1) [right = 10pt of c] {};
\node[vertex](2) [below =of c] {};
\path (0) edge[thick] node[label=above:$\tt r$] {} (1);
\path (0) edge[dotted] node[label=below:$\tt i$] {} (2);
\path (1) edge[dotted] node[label=below:$\tt i$] {} (2);
\end{tikzpicture}
\caption{${\tt r} \sqsubseteq {\tt i}\fcmp{\tt i}\converse{}$}
\end{subfigure}
\begin{subfigure}[b]{0.19\textwidth}\centering
\begin{tikzpicture}[-latex,thin]
\node[](c) {};
\node[vertex,inner sep=2pt](0) [left = 10pt of c] {};
\node[vertex,inner sep=2pt](1) [right = 10pt of c] {};
\node[vertex,inner sep=2pt](2) [below =of c] {};
\path (0) edge[dotted] node[label=above:$\tt r$] {} (1);
\path (0) edge[thick] node[label=below:$\tt i$] {} (2);
\path (1) edge[thick] node[label=below:$\tt i$] {} (2);
\end{tikzpicture}
\caption{${\tt i}\fcmp{\tt i}\converse{} \sqsubseteq {\tt r} $}
\end{subfigure}
\begin{subfigure}[b]{0.19\textwidth}\centering
\begin{tikzpicture}[-latex,thin]
\node[](c) {};
\node[vertex,inner sep=2pt](0) [left = 10pt of c] {};
\node[vertex,inner sep=2pt](1) [right = 10pt of c] {};
\node[vertex,inner sep=2pt](2) [below =of c] {};
\path (0) edge[dotted] node[label=above:$\ident{}$] {} (1);
\path (2) edge[thick] node[label=below:$\tt i$] {} (0);
\path (2) edge[thick] node[label=below:$\tt i$] {} (1);
\end{tikzpicture}
\caption{${\tt i}\converse{}\fcmp{\tt i} \sqsubseteq \ident{} $}
\end{subfigure}
\begin{subfigure}[b]{0.19\textwidth}\centering
\begin{tikzpicture}[-latex,thin]
\node[](c) {};
\node[vertex,inner sep=2pt](0) [left = 10pt of c] {};
\node[vertex,inner sep=2pt](1) [right = 10pt of c] {};
\node[vertex](2) [below =of c] {};
\path (0) edge[thick] node[label=above:$l$] {} (1);
\path (0) edge[dotted] node[label=below:$\tt r$] {} (2);
\path (2) edge[dotted] node[label=below:$\tt r$] {} (1);
\end{tikzpicture}
\caption{${\tt l} \sqsubseteq {\tt r}\fcmp{\tt r}$}
\end{subfigure}
\begin{subfigure}[b]{0.19\textwidth}\centering
\begin{tikzpicture}[-latex,thin]
\node[vertex,inner sep=2pt](0) {}; 
\node[vertex,inner sep=2pt](1) [right =of 0]{}; 
\path (0) edge[thick,bend left=90] node[label={[label distance=-3]$l$}] {} (1);
\path (0) edge[thick] node[label={[label distance=-3]$\tt r$}] {} (1);
\path (0) edge[dotted,bend right=90] node[label={[label distance=-3]below:$\bot$}] {} (1);
\end{tikzpicture}
\caption{$l \sqcap {\tt r} \sqsubseteq \bot$}
\end{subfigure}
\caption{\GR{}s for $\mathcal{T}'$}\label{fig:exgr}
\end{figure}

We start the procedure with $n = 2$ and $E = \{(l,0,1)\}$.
Note that per our optimizations, the self loops $(\ident{},0,0)$ and $(\ident{},1,1)$ are implicitly there, as well as all $\top$ edges.
Only one rule does not hold: $l \sqsubseteq {\tt r}\fcmp{\tt r}$, and consequently only one \gr{} is not maintained.
A pushout step for it gives $n = 3$ and $E = \{(l,0,1), ({\tt r},0,2), ({\tt r},2,1)\}$ as the next call to ProduceChain.
Again only one rule is not maintained, the one for ${\tt r} \sqsubseteq {\tt i}\fcmp{\tt i}\converse{}$.
This time, our work-list contains two elements: one for each edge labeled $\tt r$.
Both have a maximum node number of $2$, so we can choose either.
We pick $f$ that maps to $({\tt r},0,2)$, and $n = 4$ and $E = \{(l,0,1), ({\tt r},0,2), ({\tt r},2,1), ({\tt i},0,3), ({\tt i},2,3)\}$.
This time, ${\tt i}\fcmp{\tt i}\converse{} \sqsubseteq {\tt r}$ is also not maintained: $({\tt i},0,3)$ but $({\tt r},0,0)$ is missing.
The highest node number assigned to $N$ is $3$ however, so we need to finish treating ${\tt r} \sqsubseteq {\tt i}\fcmp{\tt i}\converse{}$.
Next iteration: $n = 5$ and $E = \{(l,0,1), ({\tt r},0,2), ({\tt r},2,1), ({\tt i},0,3), ({\tt i},1,4), ({\tt i},2,3), ({\tt i},2,4)\}$.
Subsequently: $n = 5$ and $E = \{(l,0,1), ({\tt r},0,2), ({\tt r},2,1), ({\tt i},0,3), ({\tt i},1,4), ({\tt i},2,3), ({\tt i},2,4), ({\tt r},0,0)\}$, then $({\tt r},2,2)$ is added.
At this point, we have a choice again, between ${\tt i}\fcmp{\tt i}\converse{} \sqsubseteq {\tt r}$ and ${\tt i}\converse{}\fcmp{\tt i} \sqsubseteq \ident{}$.
We apply the former first: after several iterations it gives us the graph that satisfies all rules except ${\tt i}\converse{}\fcmp{\tt i} \sqsubseteq \ident{}$:
\begin{align*}
\sem{l}G &= \{(0,1)\} \\
\sem{\tt i}G &= \{(0,3),(1,4),(2,3),(2,4)\} \\
\sem{\tt r}G &= \{(0,0),(0,2),(1,1),(1,2),(2,0),(2,1),(2,2)\}
\end{align*}
Since we did not use ${\tt i}\converse{}\fcmp{\tt i} \sqsubseteq \ident{}$ yet, and all other rules are satisfied up to this point, we are exactly in the place we would have been if ${\tt i}\converse{}\fcmp{\tt i} \sqsubseteq \ident{}$ wasn't present.
This is (minus the $l$) the graph given in Section~\ref{sec:background} as a possible graph our algorithm could give.
If we would have handled $({\tt r},2,1)$ before $({\tt r},0,2)$ instead, we would have gotten a graph with a different numbering.

We now proceed by applying ${\tt i}\converse{}\fcmp{\tt i} \sqsubseteq \ident{}$.
The pushout step adds $(\ident{},3,4)$.
We have not described precisely how our optimizations proceed at this point, but we need to renumber the nodes such that $3$ and $4$ are identified.
For preserving fairness, we renumber high to low: the node $4$ is relabeled to $3$.
This can cause some pushout steps to get assigned a lower $N$, but never a higher one.
We proceed with the graph $G'$:
\begin{align*}
\sem{l}{G'} &= \{(0,1)\} \\
\sem{\tt i}{G'} &= \{(0,3),(1,3),(2,3)\} \\
\sem{\tt r}{G'} &= \{(0,0),(0,2),(1,1),(1,2),(2,0),(2,1),(2,2)\}
\end{align*}
At this point, ${\tt i}\fcmp{\tt i}\converse{} \sqsubseteq {\tt r}$ does not hold, and the resulting action is to insert $({\tt r},0,1)$.
Subsequently, $l \sqcap {\tt r} \sqsubseteq \bot$ does not hold and we insert a conflict.
We abort concluding that the entailment holds.

While we needed several iterations to conclude entailment, we saved many iterations by treating the standard rules separately.
If we had applied ${\tt i}\converse{}\fcmp{\tt i} \sqsubseteq \ident{}$ earlier, we would have derived the contradiction sooner.

\subsection{Presentation of the Algorithm}
We conclude this section with a note on the presentation in this paper.
In the earlier paper, we presented the efficient implementation~\cite{amperspiegelRAMICS} as discussed in the previous paragraph.
This does not allow us to talk about the limit of the procedure.
Using the same presentation would have alleviated the need for Lemma~\ref{lem:identrules}.
However, the simpeler presentation used in this paper allows us to argue that the limit of a chain always exists.
This simplifies many of the other proofs in this paper.

We give an example that shows why it is problematic to describe limits in the more involved presentation:
Given the \gr{}s $\mathcal{E}\left(l \sqsubseteq l\fcmp l\right)$, $\mathcal{E}\left(l \sqsubseteq \ident{}\right)$ and the identity-rules, Figure~\ref{fig:wpcloop} shows a part of a \wpc{}.
Following the procedure for the given rules, we obtain the graphs in Figure~\ref{fig:wpcstep1} and~\ref{fig:wpcstep2}.
After every step, we could decide to apply the identity-rules until they are maintained.
If we construct a chain like this, and proceed in a similar manner as illustrated in Figure~\ref{fig:wpcloop}, we indeed construct a fair \wpc{}.
The limit of this chain is an infinite graph in which every two vertices are connected by an edge labeled $l$, as well as an edge labeled $\ident{}$, which indeed maintains the \gr{}s.
If we apply the mentioned optimizations and choose a representation of the graph as intended in Lemma~\ref{lem:identrules} after each graph, defining the `limit' becomes problematic:
We do not need to draw edges with the label $\ident$, as they are given by the drawn vertices, and the graph representation after the step in Figure~\ref{fig:wpcstep2} is drawn in Figure~\ref{fig:wpcend}.
This graph is isomorphic to the one we started with, showing we end up in a sequence that alternates between two graphs.
None of these graphs maintains any of the given \gr{}s, despite the `underlying' chain being fair.
Since a well defined limit is an important concept in many lemmas, we chose to use chains as described in this paper.

\begin{figure}[tbp]
\centering
\begin{subfigure}[b]{0.19\textwidth}\centering
\begin{tikzpicture}[-latex,semithick]
\node[vertex,label=above:$a$](1) {};
\node[vertex,label=above:$b$](2) [right =of 1] {};
\node[vertex,label=below:$c$](3) [below =of 2] {};
\path (1) edge[bend left] node[label=above:$l$] {} (2);
\path (2) edge[bend left] node[below=12pt,label=:$l$] {} (1);
\path (3) edge[bend left] node[label=below:$l$] {} (1);
\path (2) edge node[label=right:$l$] {} (3);
\path (1) edge [loop left,in=135,out=-135,distance=1cm] node[label=below:$l$] {} (1);
\end{tikzpicture}
\caption{Starting graph}\label{fig:wpcstart}
\end{subfigure}
\begin{subfigure}[b]{0.3\textwidth}\centering
\begin{tikzpicture}[-latex,thin]
\node[vertex,label=above:$a$](1) {}; 
\node[vertex,label=above:$b$,inner sep=2pt](2) [right =of 1] {};
\node[vertex,label=below:$c$,inner sep=2pt](3) [below =of 2] {};
\path (1) edge[bend left] node[label=above:$l$] {} (2);
\path (2) edge[bend left] node[below=12pt,label=:$l$] {} (1);
\path (3) edge[bend left] node[label=below:$l$] {} (1);
\path (2) edge[thick] node[label=right:$l$] {} (3);
\path (1) edge [loop left,in=135,out=-135,distance=1cm] node[label=below:$l$] {} (1);
\node[vertex,label=below:$d$](x) [right=of 2] {};
\path (2) edge[dotted] node[label=above:$l$] {} (x);
\path (x) edge[dotted] node[label=below:$l$] {} (3);
\end{tikzpicture}
\caption{First step}\label{fig:wpcstep1}
\end{subfigure}
\begin{subfigure}[b]{0.3\textwidth}\centering
\begin{tikzpicture}[-latex,thin]
\node[vertex,label=above:$a$,inner sep=2pt](1) {};
\node[vertex,label=above:$b$](2) [right =of 1] {};
\node[vertex,label=below:$c$,inner sep=2pt](3) [below =of 2] {};
\path (1) edge[bend left] node[label=above:$l$] {} (2);
\path (2) edge node[below=10pt,label=:$l$] {} (1);
\path (3) edge[bend left,thick] node[label=below:$l$] {} (1);
\path (3) edge[dotted] node[below right=6pt,label=$\ident$] {} (1);
\path (2) edge node[below right=2pt,label=:$l$] {} (3);
\path (1) edge [loop left,in=135,out=-135,distance=1cm] node[label=below:$l$] {} (1);
\node[vertex,label=below:$d$](x) [right=of 2] {};
\path (2) edge node[label=above:$l$] {} (x);
\path (x) edge node[label=below:$l$] {} (3);
\end{tikzpicture}
\caption{Second step}\label{fig:wpcstep2}
\end{subfigure}
\begin{subfigure}[b]{0.19\textwidth}\centering
\begin{tikzpicture}[-latex,semithick]
\node[vertex,label=above:$a$](1) {};
\node[vertex,label=above:$b$](2) [right =of 1] {};
\node[vertex,label=below:$d$](3) [below =of 2] {};
\path (1) edge[bend left] node[label=above:$l$] {} (2);
\path (2) edge[bend left] node[below=12pt,label=:$l$] {} (1);
\path (3) edge[bend left] node[label=below:$l$] {} (1);
\path (2) edge node[label=right:$l$] {} (3);
\path (1) edge [loop left,in=135,out=-135,distance=1cm] node[label=below:$l$] {} (1);
\end{tikzpicture}
\caption{Standardized}\label{fig:wpcend}
\end{subfigure}
\caption{Some \wps{}s}\label{fig:wpcloop}
\end{figure}

\section{A Proof of Undecidability}\label{sec:undecidable}
\begin{lemma}
The following decision problem is undecidable:
given a set of \tr{}s $\mathcal{T}$, is there a standard graph $G$ in which every \tr{} in $\mathcal{T}$ holds?
\end{lemma}
\begin{proof}
This proof closely follows a proof by Krisnadhi and Lutz~\cite{Krisnadhi:2007aa} on `conjunctive query answering'.
We use a reduction from the undecidable problem whether two context free grammars have an empty intersection.
This problem is given by two grammars with non-terminals $N_1$ and $N_2$, a common set of terminals $T$, and production rules $P_i \subseteq N_i\times(N_i \cup T)^*$ with $i\in\{1,2\}$.
The sets $N_1, N_2$ and $T$ are mutually disjoint.
The question to be answered is whether there exists a sequence of terminals that is generated by the two starting nodes $s_i \in N_i$.

We make an encoding by choosing $\mathcal{C},\mathcal{L}$ and $\mathcal{T}$ such that there is a standard graph in which the \tr{}s in $\mathcal{T}$ hold if and only if the context free grammars have an empty intersection.
We encode every terminal and nonterminal with corresponding relation symbols, and use the constant symbol $\epsilon$ for the empty word: $\mathcal{C} = \{\epsilon\}$ and $\mathcal{L} = \{\epsilon,\top,\bot,\ident{}\} \cup N_1 \cup N_2 \cup T$.
For $\mathcal{T}$ we use a \tr{} for each production rule, one for each terminal, and a final \tr{} that requires the two grammars to have an empty intersection:
\begin{multline*}
\mathcal{T} = \left\{\sigma_0 \fcmp \cdots \fcmp \sigma_k \sqsubseteq n_j \mid (n_j,\sigma_0\cdots\sigma_k) \in P_1 \cup P_2 \right\}
\cup
\left\{\ident{} \sqsubseteq t \fcmp t\converse{} \mid t \in T \right\} \\
\cup
\left\{\epsilon \fcmp s_1 \fcmp s_2\converse \fcmp \epsilon \sqsubseteq \bot \right\}
\end{multline*}

We show that there is a standard graph in which $\mathcal{T}$ holds if the grammars have an empty intersection, and there is no such graph if the grammars share a word.
First suppose the grammars have an empty intersection.
We construct a graph as follows:
The vertices are words over $T$ where $\epsilon$ is the empty word.
There are edges $(t,u,u t)$ for each non-terminal $t\in T$, edges $(n,u,u p)$ if the word $p$ is a valid parse of $n \in N_1 \cup N_2$, and edges to make the graph standard:
\begin{multline*}
E = \left\{(t,u,u t) \mid u \in T^*,\ t \in T\right\} \cup \{(\epsilon,\epsilon,\epsilon)\} \cup \{(\ident{},u,u) \mid u\in T^*\}  \\ \cup \{(\top,u,v) \mid u,v\in T^*\}
 \cup \left\{(n,u,u p) \mid u,p \in T^*,\ n \in N_1 \cup N_2,\ n\text{ parses }p\right\}
\end{multline*}
It can be checked that $G=(\mathcal{L},T^*,E)$ is standard and all \tr{}s in $\mathcal{T}$ hold.
In particular, $\sem{s_1 \fcmp s_2}G = \{\}$ as there is no word $p$ such that $s_1$ parses $p$ and $s_2$ parses $p$.

Now suppose for a proof by contradiction that $G=(\mathcal{L},V,E)$ is a standard graph in which all \tr{}s in $\mathcal{T}$ hold, and that $w = t_0 \cdots t_{n-1}$ is a word that is parsed by $s_1$ and $s_2$.
Since $\mathcal{T}$ holds, there is a path in $G$ with vertices $\epsilon=v_0,\ldots,v_n$ and edges $(t_0,v_0,v_1), \ldots, (t_{n-1},v_{n-1},v_n)$.
By induction on the parse-tree of how $s_1$ parses $w$, there is an edge $(s_1,v_0,v_n)\in E$.
Similarly, $(s_2,v_0,v_n)\in E$.
Since $\epsilon = v_0$ and $G$ is standard, $(\epsilon,\epsilon) \in \sem{\epsilon\fcmp s_1\fcmp s_2\converse\fcmp\epsilon}G$.
Since all \tr{}s in $\mathcal{T}$ hold, $(\epsilon,\epsilon) \in \sem{\bot}G$, contradicting that $G$ is standard.
\qed
\end{proof}

Undecidability of entailment follows as a corollary: there are no standard graphs in which every sentence in $\mathcal{T}$ holds iff the sentence $\top = \bot$ is entailed.

We end with a final remark about the proof we presented.
The relation symbol $\top$ and the operation $\sqcap$ were not used in the proof.
Thus, this proof of undecidability holds if $\mathcal{T}$ is restricted to \tr{}s of a simpler shape.
By application of Lemma~\ref{lem:existsgraph}, we conclude that deciding whether a conflict free consequence graph under $\mathcal{R}$ exists is undecidable.
Equivalently, given $\mathcal{R}$, it is undecidable to determine whether any least consequence graph under $\mathcal{R}$ contains a conflict.

\section{Conclusion}\label{sec:conclusion}
In this paper, we have given a translation of \tr{}s into \gr{}s, and have proven that for a graph $G$, a \tr{} is maintained in $G$ if and only if the translated \gr{} holds in $G$.
Furthermore, when allowing the \tr{}s to use extra relation-symbols with a dedicated meaning, we can add corresponding \gr{}s to ensure that this dedicated meaning is preserved.
In addition, we showed that there exists a least consequence graph of a set of \gr{}s, which one may be able to find through a semi-decision procedure.
This procedure also allows us to determine whether a set of \tr{}s is consistent.
Finally, we have shown that in a sense, we cannot do better than to give a semi-decision procedure:
The problem of whether a set of \tr{}s has a standard graph in which all \gr{}s hold is undecidable.

Our procedure can partially automate preserving invariants in information systems.
Its implementation and evaluation is foreseen in Ampersand, but considered outside the scope of this paper.


\paragraph*{Acknowledgements}
We thank Wolfram Kahl and Stef Joosten for their thoughts and comments on this paper and earlier versions.
Part of the research presented in this paper was performed at the University of Innsbruck, Austria, supported by the Austrian Science Fund (FWF) project~Y757.
Supported by The Netherlands Organisation for Scientific Research (NWO) project 639.023.710.

\bibliographystyle{splncs03}
\bibliography{Amperspiegel} 

\end{document}